\documentclass[aps,prl,twocolumn,floatfix,amsmath,ams,showpacs]{revtex4}
\usepackage{graphicx}
\usepackage{amsmath,amssymb}
\usepackage{amsfonts}
\usepackage{mathrsfs}
\usepackage{epsfig}
\usepackage{dcolumn}
\usepackage{enumerate}
\usepackage{amsthm}

\begin{document}

\newcommand{\newc}{\newcommand}
\newtheorem{proposition}{Proposition}
\newc{\beq}{\begin{equation}}
\newc{\eeq}{\end{equation}}
\newc{\kt}{\rangle}
\newc{\br}{\langle}
\newc{\beqa}{\begin{eqnarray}}
\newc{\eeqa}{\end{eqnarray}}
\newc{\pr}{\prime}
\newc{\longra}{\longrightarrow}
\newc{\ot}{\otimes}
\newc{\rarrow}{\rightarrow}
\newc{\h}{\hat}
\newc{\bom}{\boldmath}
\newc{\btd}{\bigtriangledown}
\newc{\al}{\alpha}
\newc{\be}{\beta}
\newc{\ld}{\lambda}
\newc{\sg}{\sigma}
\newc{\p}{\psi}
\newc{\eps}{\epsilon}
\newc{\om}{\omega}
\newc{\mb}{\mbox}
\newc{\tm}{\times}
\newc{\hu}{\hat{u}}
\newc{\hv}{\hat{v}}
\newc{\hk}{\hat{K}}
\newc{\ra}{\rightarrow}
\newc{\non}{\nonumber}
\newc{\ul}{\underline}
\newc{\hs}{\hspace}
\newc{\longla}{\longleftarrow}
\newc{\ts}{\textstyle}
\newc{\f}{\frac}
\newc{\df}{\dfrac}
\newc{\ovl}{\overline}
\newc{\bc}{\begin{center}}
\newc{\ec}{\end{center}}
\newc{\dg}{\dagger}
\newc{\prh}{\mbox{PR}_H}
\newc{\prq}{\mbox{PR}_q}
\newc{\tr}{\mbox{tr}}
\newc{\pd}{\partial}
\newc{\qv}{\vec{q}}
\newc{\pv}{\vec{p}}
\newc{\dqv}{\delta\vec{q}}
\newc{\dpv}{\delta\vec{p}}
\newc{\mbq}{\mathbf{q}}
%\newc{\dmbq}{\mathbf{\delta q}}
\newc{\mbqp}{\mathbf{q'}}
\newc{\mbpp}{\mathbf{p'}}
\newc{\mbp}{\mathbf{p}}
\newc{\mbn}{\mathbf{\nabla}}
\newc{\dmbq}{\delta \mbq}
\newc{\dmbp}{\delta \mbp}
\newc{\T}{\mathsf{T}}
\newc{\J}{\mathsf{J}}
\newc{\sfL}{\mathsf{L}}
\newc{\C}{\mathsf{C}}
\newc{\B}{\mathsf{M}}
\newc{\V}{\mathsf{V}}
%\newc{\det}{\mbox{Det}}

% The notion of entanglement is a purely quantum mechanical phenomenon that is absent in classical physics. 
% It is studied since Schr\"odinger and the famous paper of Einstein, Podolsky and Rosen (EPR) 
% \cite{epr_paradox}, correlations of which seem to imply nonlocality. The work of Bell \cite{bell} and 
% others led to famous inequalities that quantified the extent to which classical correlations can be surpassed.
%  These inequalities were experimentally verified by Aspect {\it et.al} \cite{aspect}.

\title{Using Partial Transpose and Realignment to generate Local Unitary Invariants}

% \author{Udaysinh T. Bhosale \footnote{e-mail: bhosale@physics.iitm.ac.in}}
% % \affiliation{Department of Physics, Indian Institute of Technology Madras, Chennai, 600036, India}
% \author{K. V. Shuddhodan \footnote{e-mail: kvshud@math.tifr.res.in}}
% % \affiliation{}
% \altaffiliation{School of Mathematics, Tata Institute of Fundamental Research, Homi Bhabha
% Road, Bombay 400005, India}
% \author{Arul Lakshminarayan \footnote{e-mail: arul@physics.iitm.ac.in}}
% \affiliation{Department of Physics, Indian Institute of Technology Madras, Chennai, 600036, India}

\author{Udaysinh T. Bhosale}
\email{bhosale@physics.iitm.ac.in}
\author{K. V. Shuddhodan}
%\email{kvshud@math.tifr.res.in}
\altaffiliation{Present address: School of Mathematics, Tata Institute of Fundamental Research, Homi Bhabha
Road, Mumbai 400005, India.}
\author{Arul Lakshminarayan}
%\email{arul@physics.iitm.ac.in}
\affiliation{Department of Physics, Indian Institute of Technology Madras, Chennai, 600036, India}
\date{\today}

\preprint{IITM/PH/TH/2010/11}
\begin{abstract}
 Motivated by link transformations of lattice gauge theory, a method for generating local unitary invariants, especially for a system of qubits, has been pointed out in an earlier work [M. S. Williamson {\it et. al.}, Phys. Rev. A {\bf 83}, 062308 (2011)].  This paper first points the equivalence of the so constructed transformations to the combined operations of partial transpose and realignment. This allows construction of local unitary invariants of any system, with subsystems of arbitrary dimensions. Some properties of the resulting operators and consequences for pure tripartite higher dimensional states are briefly discussed.

\end{abstract}
\pacs{03.67.Bg, 03.67.Mn}

\maketitle

Entanglement, the remarkable nonlocal feature of quantum mechanics, has been extensively studied in the recent past \cite{Horodeckirpm}. Owing mainly due to its role in quantum information protocols,  
for example in teleportation \cite{Teleport}, dense coding \cite{Superdense}, and channel discrimination \cite{Pianiwatrous}, as well as due to its 
presence in quantum algorithms that have a speedup over classical ones \cite{Jozsalinden}.
It also plays an increasingly prominent role in condensed matter physics \cite{Amico2008}, and presumably also affects macroscopic observables such as magnetic properties in some solids \cite{Ghosh}.  

However, to detect and measure entanglement in a general state, represented by a density matrix, is a difficult problem. 
It has been solved for the case of two qubit states \cite{Wooters,Wootersentform,mhorodecki}, bipartite states of a
qubit and a qutrit \cite{mhorodecki}, and all bipartite pure states.
There are various measures of entanglement, for example the von Neumann entropy of any one subsystem in a bipartite pure state
is such a measure, the concurrence introduced in \cite{Wooters,Wootersentform} measures entanglement between 
two qubits in a pure or mixed state, while negativity and log-negativity \cite{vidal,logneg} is invoked for a general bipartite mixed state which uses 
the positive, but not completely positive, map of partial transposition \cite{Peres}. The partial transpose was introduced to detect entanglement,
and provides a sufficient but not necessary condition. One other such criterion uses the operation of ``realignment" to be expanded on further below.

Another approach to study entanglement, especially of a multipartite kind, is by studying the local unitary (LU) invariants of the system 
which consists of invariants under arbitrary unitary transforms restricted to the individual subsystems. In as much as entanglement quantifies non-local properties, entanglement measures remain invariant under LU operations and hence the importance of their study. The spectra of the density matrix itself and the various reduced density matrices got by tracing out subsystems are such LU invariants. However, it helps to have invariants that are polynomials in the 
entries of the density matrix \cite{Grassl98}, and those whose physical interpretation in terms of entanglement is available. 
These invariants uniquely determine the orbit of the state under these local operations. 
For example three qubit pure states have five independent LU invariants 
\cite{Linden98,Carteret99} excluding the trace. In a recent paper Williamson {\it et.al}~\cite{Williamson11}, inspired by lattice gauge theory,
  have given a method of generating these invariants by associating them with a closed path joining some or all the 
qubits where two consecutive qubits on the path are connected by a ``link transformation''.
More recently index-free formulas for invariants of $k$ qudit system up to degree six by using graph-theoretic methods have been given \cite{Szalay12}.

In this paper, it is first shown that the link transformation in 
\cite{Williamson11} is unitarily equivalent to the combined operations of partial transpose \cite{Peres} and realignment \cite{Chen03}  
[refer Eq.~(\ref{mainresult})]. Thus it is interesting that two rather independent operations on which entanglement criteria are based
 come together in the construction of local invariants. 
This result immediately suggests a way to generalize to a system of an arbitrary set of qudits,
each not necessarily of identical dimensions.  Such a generalization is then shown to be LU invariant as well.
One of the advantages of this method is that it does not need the generalization of Pauli matrices in higher 
dimensions to get the invariants as required for the link transformation approach of \cite{Williamson11}.

For the present, it helps to first summarize the central result of \cite{Williamson11}. Consider a lattice of points 
$a_1, \, a_2,\,.~.~.~,\, a_N$ where each point represents a qubit and $N$ is the total number of qubits.
The invariants are generated by various closed path (containing some or all the qubits) and two 
consecutive qubits are connected by a link transformation. For example consider a closed path as shown
in Fig.~1 and qubits $a_1$ and $a_2$ are lying on it then the link transformation 
connecting them is denoted by $S(a_2,a_1)$, %(refer Fig.~(\ref{pathdigr})) 
whose elements are given by
\beq
\label{linkmatrix1}
S(a_2,a_1)_{nm}=\dfrac{1}{2} \mbox{tr}[\rho_{12}\; \sigma^{a_1}_m \otimes \sigma^{a_2}_n].
\eeq

Here $\rho_{12}$ is the reduced density matrix of qubits $a_1$ and $a_2$,
$\sigma_l$  ($l=0, 1, 2, 3$) are $2\times2$ matrices such that $l=0$ corresponds to the identity matrix and
$l=1,2,3$ corresponds to the standard Pauli matrices $\sigma_x, \sigma_y, \sigma_z$ respectively. In a similar way consecutive qubits are connected on the path by corresponding transformation matrices. The interpretation of $S(a_2,a_1)$ was provided as ${\mathbf m}_2 = S(a_2,a_1) {\mathbf m}_1$, where ${\mathbf m}_1$ represented an arbitrary measurement operator in the space of subsystem $a_1$, while
 ${\mathbf m}_2$  was the resultant state of subsystem $a_2$, all of these being written in the $\{\sigma_l\}$ basis. 

Then for a given closed path, for example consisting of $K$ qubits, not necessarily all distinct: $\{a\}_K\equiv (a_1, \, a_2,\, \ldots,\, a_K)$ , the quantity
\beq
\label{smatrix}
\mbox{tr}[ S(a_1,a_K) \cdots S(a_3,a_2) S(a_2,a_1)]
\eeq
was shown to be a LU invariant. Note that this expression is to be read from right to left. Different path coordinates $a_i$ can refer to the same subsystem, depending on the actual path. One can obtain invariants by taking various closed paths. In \cite{Williamson11} a maximal set of five independent invariants of three qubit pure states were calculated in this manner. 

It is now shown that realignment of a partial transpose is a unitary transform of the link transformation $S(a_2,a_1)$.
In the partial transpose (PT) operation, transposition is done only on one subsystem. As mentioned already, PT is a positive but not a completely positive map and can hence be used to detect entanglement. Consider a bipartite density matrix $\rho_{12}$,
which refers to the reduced density matrix of the subsystems at points $a_1$ and $a_2$ of the path, and any orthonormal product basis $\{|i\alpha\kt\}$ in the corresponding state space, then the transposition on only the subsystem labeled $a_2$, denoted as $\rho_{12}^{T_2}$, is given by the matrix elements:
\begin{equation}
\langle i|\langle \beta| \rho_{12}^{T_2}|j\rangle|\alpha \rangle = 
\langle i|\langle \alpha|\rho_{12}|j\rangle|\beta \rangle.
\end{equation}
The PT criterion states that if $\rho_{12}^{T_2}$ is negative then the state is entangled \cite{Peres}, but otherwise it may or may not be separable, with separability guaranteed only for $2\times2$ and $2\times3$ systems \cite{mhorodecki}. 
This implies that positivity under PT is a necessary but not sufficient condition for separability.

The other operation of interest is realignment, to which is associated what is called in literature the
{\it{computable cross norm}} (CCN) {\it{criterion}} \cite{Chen03,Oliver04}, or simply the 
realignment criterion. The corresponding operation on the density matrix $\rho_{12}$, denoted as $\mathcal{R}\left(\rho_{12}\right)$
 is given by:
\beq
\langle i|\langle j|\left(\mathcal{R}(\rho_{12})\right)|\alpha\rangle|\beta \rangle= \langle i|\langle \alpha|\rho_{12}|j\rangle|\beta \rangle.
\label{Eq:RealDef}
\eeq
The criterion is that if the state $\rho_{12}$ is separable then $\Vert \mathcal{R}(\rho_{12})\Vert_1 \le1$, where 
$\Vert M \Vert_1$ is the  trace norm equal to $\mbox{tr}\sqrt{M M^{\dagger}}$ \cite{Oliver03}. 
This condition is found to detect some bound entangled states, these being positive under PT and hence not being detected by the 
corresponding criterion \cite{Oliver04,mhorodeckibound}.

Define the map $\rho_{12} \mapsto \mathcal{R}(\rho_{12}^{T_2})$, where {\it both} PT and realignment are affected serially. 
As both of these operations are merely permutations of the matrix elements, so is this combination. It is now shown that for two qubits, 
$S(a_1,a_2)$ is unitarily equivalent to $\mathcal{R}(\rho_{12}^{T_2})$. Starting from Eq.~(\ref{linkmatrix1}) and using an arbitrary 
orthogonal product basis {$|i\beta\kt$} (where $i$, $\beta=0,1$) and Einstein's summation convention (operative throughout the paper),
one obtains the following:
%\begin{equation}
% S(a_2,a_1)_{nm}=\dfrac{1}{2} \mbox{tr}[\rho_{12} \sigma^{a_1}_m \otimes \sigma^{a_2}_n]
%\end{equation}
%
\begin{eqnarray}
\label{amnform}
\begin{split}
 S(a_2,a_1)_{nm} &= \dfrac{1}{2} \br i \beta| \rho_{12} \; \sigma^{a_1}_m \otimes \sigma^{a_2}_n |i \beta \kt \\
&= \dfrac{1}{2} \br i \beta| \rho_{12}|j\alpha\kt \br j\alpha| \sigma^{a_1}_m \otimes \sigma^{a_2}_n |i \beta \kt\\
&= \dfrac{1}{2} \br \beta \alpha|  \mathcal{R}(\rho_{21}^{T_1}) |ji\kt\br j|\sigma^{a_1}_m |i\kt\br\alpha|\sigma^{a_2}_n|\beta\kt.
\end{split}
\end{eqnarray}
% where superscript $T_2$ denotes transposition done on second subsystem and $\mathcal{R}$ is the
% realignment operation defined in \cite{Chen03}  such that 
% \begin{equation*} 
% \br i j| \mathcal{R}(\rho_{12}^{T_2})|\alpha\beta\kt=\br i \beta|\rho_{12}| j \alpha \kt \;
% \mbox{and}\; S(a_1,a_2)\; \mbox{is the link transformation given in \cite{Williamson11}}.
% \end{equation*}
Define a $4 \times 4$ matrix $U$ with elements:
\begin{equation} 
\label{defineU}
\br ji|U|m\kt=\dfrac{1}{\sqrt{2}}\br j|\sigma_m|i\kt,
\end{equation}
where $i,\,j$=0,1 and $m= 0,1,2,3$. The matrix $U$ is independent of the qubit label, and hence these have been omitted from the Pauli matrix symbols.
Thus using Eq.~(\ref{defineU}) and the last expression of Eq.~(\ref{amnform}) gives the following equation:
\begin{eqnarray*}
  S(a_2,a_1)_{nm} &=&  \br \alpha\beta|U|n\kt  \br \beta \alpha|  \mathcal{R}(\rho_{21}^{T_1}) |ji\kt  \br ji|U|m\kt.
\end{eqnarray*}
On taking the complex conjugate of  Eq.~(\ref{defineU}) and using the definition of $U$ matrix,
it can be shown that $\br  \alpha \beta |U|n \kt =\br n|U^{\dagger}|\beta \alpha \kt$. 
Thus one obtains the advertised relation between the link transformation $[S(a_2,a_1)]$, and the operations 
of  partial transpose and realignment as
\begin{equation}
\label{mainresult}
S(a_2,a_1) = U^{\dagger} \mathcal{R}(\rho_{21}^{T_1}) U.
\end{equation}
 
 The matrix $U$ written explicitly in the standard basis is
\begin{equation}
\label{matrixU}
U = \dfrac{1}{\sqrt{2}}\left( \begin{array}{llrr}
	      1 & 0 &  0 & 1\\
	      0 & 1 & -i & 0\\
	      0 & 1 &  i & 0\\     
	      1 & 0 &  0 & -1\\
\end{array}\right).
\end{equation}
%which is the identity and the three Pauli matrices deformed into 4-vectors and arranged along columns in a matrix.
Deforming the $2\times 2$ identity matrix and the three Pauli matrices into four-vectors and arranging them along columns makes the 
matrix $\sqrt{2} \, U$, the first column of which corresponds to the identity matrix while the other three correspond to the matrices
$\sigma_1$, $\sigma_2$ and $\sigma_3$ respectively. These can also be thought of as resulting from realigning the 
$\{\sigma_l,\; 0 \le l \le 3\}$ basis. Indeed any matrix which is a $(d_1 d_2 \times d_1 d_2)$-dimensional array can be realigned into 
an array of dimension $d_1^2 \times d_2^2$ using the definition in Eq.~(\ref{Eq:RealDef}). For example, the $2\times 2$ identity matrix, 
an array with $d_1=2$ and $d_2=1$ is realigned into the $4\times 1$ vector $(1,0,0,1)^T$, the first column in the matrix $\sqrt{2}\, U$, 
while $\sigma_1$ realigned into the vector $(0,1,1,0)^T$, is the second column, and so on.

It is readily seen that the matrix $U$ is indeed an unitary matrix with the additional property that $U U^T = \mathcal{S}$, the SWAP operator.
An invariant corresponding to the closed loop $(a_1, \, a_2,\, \ldots,\, a_K)$ is then
\begin{eqnarray}
\begin{split}
& \mbox{tr}[ S(a_1,a_K) \cdots S(a_3,a_2) S(a_2,a_1)] \\
&= \mbox{tr}[  
U^{\dagger} \mathcal{R}(\rho_{1K}^{T_K}) U  \cdots  U^{\dagger} \mathcal{R}(\rho_{32}^{T_2}) U U^{\dagger} \mathcal{R}(\rho_{21}^{T_1}) U ]\\
\label{equality}
&= \mbox{tr}[   \mathcal{R}(\rho_{1K}^{T_K}) \cdots  \mathcal{R}(\rho_{32}^{T_2}) \mathcal{R}(\rho_{21}^{T_1})].
\end{split}
\end{eqnarray}
Hence the link transformation matrices such as $S(a_2,a_1)$ can be replaced by $\mathcal{R}(\rho_{21}^{T_1})$, and referred to as transformation matrices themselves. That these are link transformations in their own right and are in fact amenable to the identical interpretations as $S(a_2,a_1)$ in \cite{Williamson11} follows from expressing the measurement operator at subsystem $a_1$ in the standard basis ${\mathbf e}_{ij}$, which are matrices with 1 at positions $ij$ and 0 elsewhere and $i,j \in\{1,2\}$. That is the following is the equivalent to Eq.~(\ref{linkmatrix1}):
\beq
\label{linkmatrix2}
\left(\mathcal{R}(\rho_{21}^{T_1})\right)_{ij;\alpha \beta}= \mbox{tr}\left[\rho_{12}\; {\mathbf e}^{a_1}_{\alpha \beta} \otimes \left({\mathbf e}^{a_2}_{ij}\right)^T \right].
\eeq
Generalizing beyond qubits to subsystems of higher dimensions $d_i$ (qudits) it is shown in Proposition~(\ref{prop1}) that the 
last term of Eq.~(\ref{equality}) is a LU invariant, thus obviating the need to specially generalize link transformations such as $S(a_2,a_1)$. 
The quantity $\mathcal{R}(\rho_{21}^{T_1})$ is in general a rectangular array of dimension $d_2^2 \times d_1^2$. The dimensions of the various rectangular arrays dovetail such that the final array on completion of the loop based at $1$ is a square matrix of dimension $d_1^2$.

\begin{proposition}
\label{prop1}
Under local unitary operations, $U_i$ let the transformed two-body density matrices be 
$\tilde{\rho}_{ij} =U_i \otimes U_j \rho_{ij} U_i^{\dagger} \otimes U_j^{\dagger}$. 
If $\{a\}_K \equiv (a_1, \, a_2,\, \ldots,\, a_K)$ is a closed path, then
\beq
\label{LU}
\begin{split}
 &\mathcal{P}(\{a\}_K) \equiv \mathcal{R}(\rho_{1K}^{T_K}) \cdots  \mathcal{R}(\rho_{32}^{T_2}) \mathcal{R}(\rho_{21}^{T_1})=\\
&(U_1 \otimes U_1^{*})^{\dagger}\mathcal{R}({\tilde{\rho}}_{1K}^{T_K}) \cdots  \mathcal{R}({\tilde{\rho}}_{32}^{T_2}) \mathcal{R}({\tilde{\rho}}_{21}^{T_1})
(U_1 \otimes U_1^{*}).
\end{split}
\eeq
\end{proposition}

\begin{proof} This follows from the following observation:
%\begin{subequations}
%\begin{eqnarray} 
%%\begin{split}
%{\tilde{\rho}}_{21}&=&(U_2 \otimes I_1) \rho_{21} (U_2^{\dagger}\otimes I_1) \label{21transf}\\
%\mbox{and}\; \;{\tilde{\rho}}_{32}&=&(I_3 \otimes U_2 ) \rho_{23} (I_3 \otimes U_2^{\dagger}), \label{32transf}
%%\end{split}
%\end{eqnarray}
%\end{subequations}
%\begin{subequations}
\begin{eqnarray} 
%\mbox{then}\;\;
\mathcal{R}(\rho_{21}^{T_1}) &=& (U_2 \otimes U_2^{*})^{\dagger} \mathcal{R}(\tilde{\rho}_{21}^{T_1})(U_1 \otimes U_1^{*}). \label{pt210}
%\mathcal{R}(\rho_{32}^{T_2}) &=&\mathcal{R}(\tilde{\rho}_{32}^{T_2})(U_2 \otimes U_2^{*}),\label{pt320}\\
\label{pro1eqns}
%\mathcal{R}(\rho_{21}) &=&  (U_2 \otimes U_2^{*})^{\dagger} \mathcal{R}(\tilde{\rho}_{21}) (U_1^{*} \otimes U_1).
%\label{nopt320}
%\mbox{and}\;\; \mathcal{R}(\rho_{32}) &=&  \mathcal{R}({\tilde{\rho}}_{32}) (U_2^{*} \otimes U_2).
\end{eqnarray}
%\end{subequations}

For simplicity, and without loss of generality, consider the case of only a local unitary $U_2$ acting on the second subsystem, that is
$\rho_{21}=( U_2^{\dagger}\otimes I_1) {\tilde \rho_{21}} (U_2 \otimes I_1)$.
It follows that:
\begin{equation}
(\rho_{21})_{i\alpha;j\beta} = (U_2^{\dagger})_{i,i'}
({\tilde \rho_{21}})_{i'\alpha;j'\beta} (U_2)_{j',j}.\\
 \end{equation}
Using the definitions of the realignment and partial transpose operations the above is rewritten as:
\begin{eqnarray}
\begin{split}
\left(\mathcal{R}(\rho_{21}^{T_1})\right)_{ij;\beta\alpha} &= (U_2^{\dagger})_{i,i'}
\left(\mathcal{R}({\tilde{\rho}}_{21}^{T_1})\right)_{i'j';\beta\alpha}(U_2)_{j',j}\\
&= (U_2^{\dagger})_{i,i'} (U_2^{*})^{\dagger}_{j,j'} \left(\mathcal{R}({\tilde{\rho}}_{21}^{T_1})\right)_{i'j';\beta\alpha}.
\end{split}
\end{eqnarray}
This leads to the left action (that is multiplication by $(U_2 \otimes U_2^{*})^{\dagger}$ on the left) in Eq.(\ref{pt210}).
Considering separately the case when only the local unitary $U_1$ is operative, leads to the corresponding right action  
 (multiplication by $U_1 \otimes U_1^{*}$ on the right) in Eq.(\ref{pt210}) and completes the proof. 
Thus if $U_1$, $U_2$, $\ldots$ $U_K$ are local unitary operators acting on subsystems $1$, $2$, $\ldots$, $K$ then their combined action on the successive bipartite states $\rho_{i+1i}$ leads to the claim in Eq.~(\ref{LU}).
\end{proof}

Thus it follows that the eigenvalues of $\mathcal{P}(\{a\}_K)$ are LU invariants, as well as its trace. This provides the generalization of the link transformation approach to generate local unitary invariants by using realignment and PT operations {\it i.e.} replacing $S(a_2,a_1)$ by $\mathcal{R}(\rho_{21}^{T_1})$, see Fig.~(\ref{pathdigr}). In fact the above proof does not need special properties of the Pauli matrices 
used in \cite{Williamson11}, and obviates completely the need for their generalizations and hence presents a much simpler algorithm.

\begin{figure}
\begin{center}
%        \resizebox{100mm}{!}{\includegraphics{pathdigr.eps}} 
\includegraphics[scale=0.7]{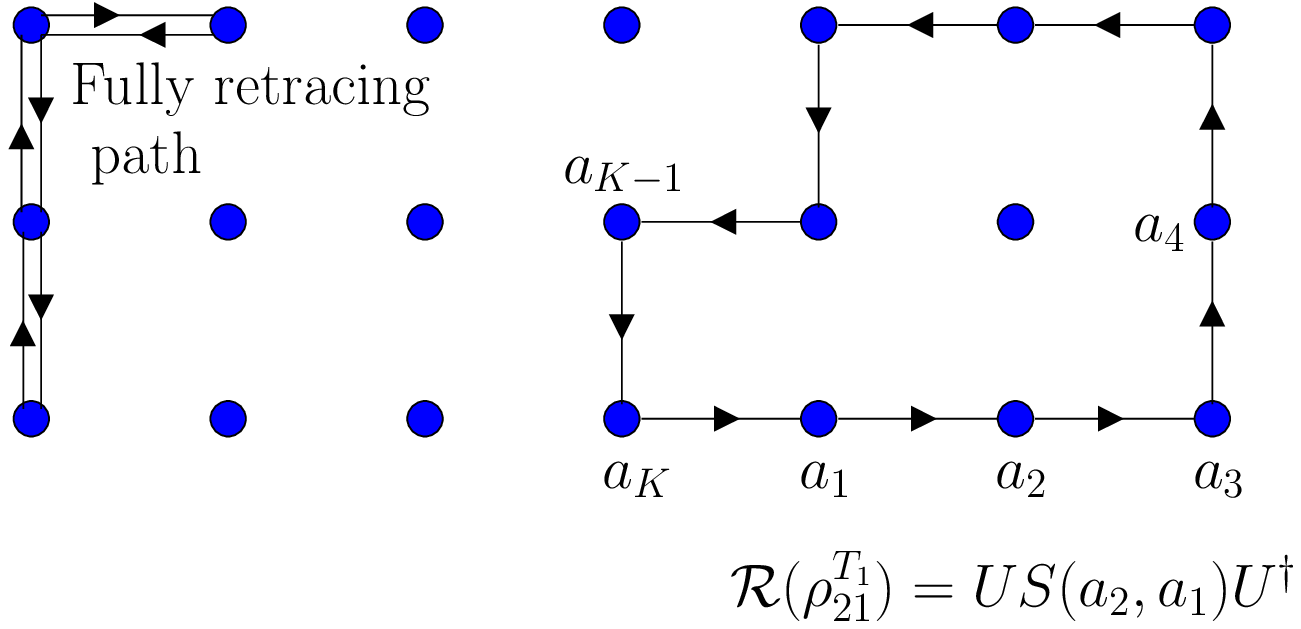}
\caption{(Color online) Closed path connecting some qudits of a total of $K$ qudits is shown.
% Qudits $a_1$ and $a_2$ are connected by transformation matrix $\mathcal{R}(\rho_{12}^{T_2})$ and 
% similarly other two consecutive qudits on the path. 
Also shown is a fully retracing path on four qudits.
Any two consecutive qudits on a path are connected by a transformation matrix, for example, 
qudits $a_1$ and $a_2$ are connected by $\mathcal{R}(\rho_{12}^{T_2})$.
The equality shown corresponds to the case of qubits.} 
\label{pathdigr} 
\end{center}
 \end{figure} 
 
 The need for the joint operations of realignment and PT is interesting and if say one uses
 only realignment, it is easy to see that this does not in general lead to LU invariants. This follows
 as 
 \beq
 \mathcal{R}(\rho_{21}) =  (U_2 \otimes U_2^{*})^{\dagger} \mathcal{R}(\tilde{\rho}_{21}) (U_1^{*} \otimes U_1).
\label{nopt320}
 \eeq
 While this is derived in a similar way as Eq.~(\ref{pro1eqns}) note the different ordering of the unitary matrices in them. 
  Thus if {\it orthogonal} transforms are used instead of the more general unitary ones, the operation of PT can be dispensed with.
 Alternatively if one restricts oneself to a real state space (such as the so-called rebit space \cite{Rungta01}) these 
 will in fact be local invariants.

Given a path $\{a\}_K$ and the associated operator $\mathcal{P}(  \{a\}_K)$, there are many different ways of 
associating LU invariants, even if one were to restrict to polynomials of the density matrix elements. For instance
the coefficients of the characteristic polynomials of the operator are such invariants and these include the trace and determinant. That such coefficients are real follows from a property that is stated as part of Proposition~(\ref{prop2}), to be proved below.
Also if the path is fully retracing (there exists a number $L$ such that the path is $\{a_1, \ldots, a_L,a_{L-1} \ldots a_1\}$, the retracing path shown in  Fig.~(\ref{pathdigr}) has $L=4$.) then $\mathcal{P}(  \{a\}_K)$ is a positive operator itself. A necessary ingredient in the
properties of the the operator is the SWAP operation. Consider the SWAP operator $\mathcal{S}_1$ which acts on the product space of two identical dimensional spaces: $\mathcal{H}_{d_1}\otimes \mathcal{H}_{d_1}$. Its action is given by   
$\langle lm|\mathcal{S}_1|ij\rangle=\delta_{lj}\delta_{mi}$.
The SWAP operator is symmetric permutation matrix such that $\mathcal{S}_1^2=I$ which implies further that it is unitary. Swap operators for other subsystems, $\mathcal{S}_i$ are similarly defined.

\begin{proposition}
\label{prop2}
If the closed path $\{a\}_K$ is fully retracing then  $\mathcal{P}(\{a\}_K)$ is a positive operator, 
else it's eigenvalues are either real or appear as complex conjugate pairs.
\end{proposition}  
\begin{proof}
This follows from two observations:
\begin{subequations}
\beqa
 \mathcal{S}_2 \mathcal{R}(\rho_{21}^{T_1}) \mathcal{S}_1 &=& \mathcal{R}(\rho_{21}^{T_1})^*,\label{SRS}\\
 \mathcal{R}(\rho_{21}^{T_1})&=& {\mathcal{R}(\rho_{12}^{T_2})}^{\dagger}. \label{RRd}
\eeqa
\end{subequations}
 Considering  a matrix element of $\mathcal{S}_2 \mathcal{R}(\rho_{21}^{T_1}) \mathcal{S}_1$:
\begin{eqnarray}
\begin{split}
&\left(\mathcal{S}_2 \mathcal{R}(\rho_{21}^{T_1}) \mathcal{S}_1 \right)_{ij;kl} =
 (\mathcal{S}_2)_{ij;pq}(\mathcal{R}(\rho_{21}^{T_1}))_{pq;rs} (\mathcal{S}_1)_{rs;kl} \\
&= \delta_{iq} \delta_{jp} (\mathcal{R}(\rho_{21}^{T_1}))_{pq;rs}   \delta_{rl} \delta_{sk}
= (\mathcal{R}(\rho_{21}^{T_1}))_{ji;lk}\\
&= (\rho_{21})_{jk;il}
= (\rho_{21}^{*})_{il;jk} 
= (\mathcal{R}(\rho_{21}^{T_1}))_{ij;kl}^{*},
\end{split}
\end{eqnarray}
where the second last equality follows from the hermiticity of $\rho_{21}$.
This proves Eq.(\ref{SRS}), using which and inserting $S_i^2$, $i=2, \ldots, K$ as shown below, 
leads to   
\[
\begin{split}
&\mathcal{S}_1\mathcal{P}(\{a\}_K)\mathcal{S}_1=  \mathcal{S}_1 \mathcal{R}(\rho_{1K}^{T_K}) \mathcal{S}_K \cdots \mathcal{S}_3 \mathcal{R}(\rho_{32}^{T_2}) \mathcal{S}_2 
\mathcal{S}_2 \mathcal{R}(\rho_{21}^{T_1}) \mathcal{S}_1 \\
& =\left(\mathcal{R}(\rho_{1K}^{T_K})\right)^{*} \cdots  \left(\mathcal{R}(\rho_{32}^{T_2})\right)^{*} 
 \left(\mathcal{R}(\rho_{21}^{T_1})\right)^{*}  
 =\mathcal{P}(\{a\}_K)^*.
\end{split}
\]
In other words $\mathcal{P}(\{a\}_K)$ is unitarily equivalent to its complex conjugate through the SWAP $\mathcal{S}_1$. Thus it immediately follows that $\mathcal{P}(\{a\}_K)$ has a real characteristic polynomial and the eigenvalues appear as stated.

\begin{figure}
\begin{center}
%        \resizebox{100mm}{!}{\includegraphics{pathdigr.eps}} 
\includegraphics[scale=0.35]{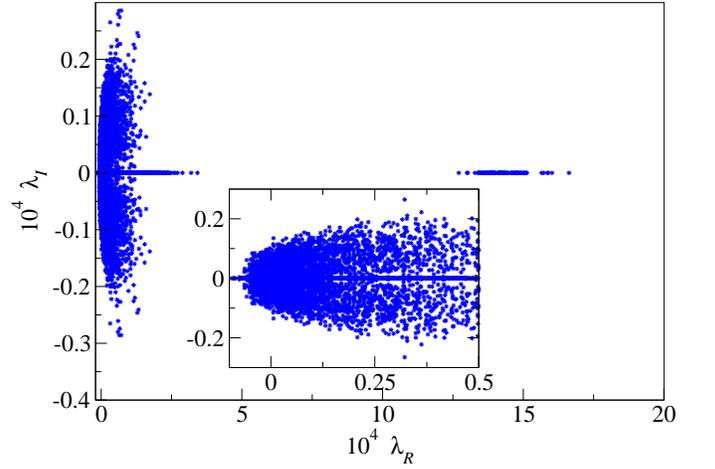}
\caption{(Color online) Real part ($\lambda_R$) and imaginary part ($\lambda_I$) of the eigenvalues of the matrix corresponding to  
$\mathcal{P}(\{1,2,3\})$ for 100 random tripartite pure states with subsystem dimensions $d_1=d_2=d_3=10$. The states are sampled according to the uniform Haar measure.  The inset show an enlarged view of the region near the origin.} 
\label{kempefig} 
\end{center}
 \end{figure} 
 
That if the path is exactly retracing, $\mathcal{P}(\{a\}_K)$ is a positive operator follows from
Eq.(\ref{RRd}) which in turn follows from:
\begin{eqnarray}
\begin{split}
&\left(\mathcal{R}(\rho_{21}^{T_1})\right)_{ij;\alpha\beta} = (\rho_{21})_{i \beta;j \alpha}
= (\rho_{21})_{j \alpha; i \beta}^{*}
= (\rho_{12})_{\alpha j;\beta i}^{*}\\
&= \left(\mathcal{R}(\rho_{12}^{T_2})\right)_{\alpha\beta;ij}^{*}
= \left(\mathcal{R}(\rho_{12}^{T_2})\right)_{ij;\alpha\beta}^{\dagger}.
\end{split}
\end{eqnarray}
\end{proof}
The closed path connecting all subsystems of a pure state in one loop maybe of special interest. 
The simplest case is that of bipartite states and the quantity $\det[\mathcal{R}(\rho_{12}^{T_2})\mathcal{R}(\rho_{21}^{T_1})]^{1/4}$ 
for the case of a two qubit pure state is $\tau/4$, where $\tau$ is the two-tangle \cite{Coffman} i.e. square of the concurrence. 

Finally, some consequences for tripartite higher dimensional pure states are discussed. 
For example for three qubits, $\mbox{tr}[\mathcal{P}(\{1,2,3\})]$ over such a path is the Kempe 
invariant \cite{Kempe99,Williamson11}. Its generalization to qudits then is evident from the formulation 
present above and is $\mbox{tr}[\mathcal{P}(\{1,\ldots,N\})]$ valid for $N$ subsystems of arbitrary 
dimensions. For example the eigenvalues of 
$\mathcal{P}(\{1,2,3\})= \mathcal{R}(\rho_{13}^{T_3}) \mathcal{R}(\rho_{32}^{T_2}) \mathcal{R}(\rho_{21}^{T_1})$ 
are shown in Fig.~(\ref{kempefig}) for the case of random pure tripartite states of three systems each of 
dimension 10. It is observed that one real eigenvalue per realization is significantly larger in magnitude compared to others. 
This is a robust feature that is present in many operators, including elemental ones such as $ \mathcal{R}(\rho_{21}^{T_1})$. 
The origin of this is not hard to understand if one assumes a diagonally dominant density matrix, as typical ones indeed are. 
For a random pure tripartite state in a space of dimension $d_1\times d_2\times d_3$, the diagonal dominance of the density matrix $\rho_{21}$ follows from the fact that the modulus of the diagonal elements is of order $1/(d_1 d_2)$ and that of the off-diagonal elements is $1/(d_1 d_2\sqrt{d_3})$ \cite{Arulentpow}.

It is also not hard to see that when the state $|\psi_{123}\kt$ is bi-separable, in particular $|\psi_{123}\kt=|\phi_{12}\kt\otimes |\chi_{3}\kt$, the eigenvalue spectrum of $\mathcal{P}(\{1,2,3\})$ consists of only one non-zero real value equal to $\mbox{tr}[(\rho_{12}^{T_2})^3]$. For the path $1\rarrow 2 \rarrow 1 \rarrow 2$, the same state gives for another invariant $I_6' \equiv \mbox{tr}(\mathcal{P}(\{1,2,1,2\}))= \mbox{tr}[\mathcal{R}(\rho_{12}^{T_2}) \mathcal{R}(\rho_{21}^{T_1})]^2= [\mbox{tr}(\rho_1^2)]^2$, this also being the only non-vanishing eigenvalue of 
$\mathcal{P}(\{1,2,1,2\})$. Further for tri-separable pure states the eigenvalues of  $\mathcal{P}(\{1,2,3\})$ and $\mathcal{P}(\{1,2,1,2\})$ are all zero except one whose value is 1.

In conclusion this paper has pointed to a close connection between local invariants and the two operations that have hitherto been 
used to reveal entanglement, namely partial transpose and realignment. Properties of the resulting operators have been discussed and 
some simple consequences for qudit tripartite systems have been pointed out. It is hoped that this will lead to the construction of 
useful multipartite entanglement measures, and an understanding of invariants.

\acknowledgements{AL thanks Prabha Mandyam for pointing out reference \cite{Williamson11}.}

\bibliography{ref2010,refs}
\end{document}